\documentclass[twoside,11pt,titlepage]{amsart}
\usepackage{amsmath} 
\makeatletter
\newcommand{\addresseshere}{%
  \enddoc@text\let\enddoc@text\relax
}
\makeatother

\usepackage{amsthm} 
\usepackage{amssymb}	
\usepackage{graphicx} 
\usepackage{multicol} 
\usepackage[marginratio=1:1,height=584pt,width=360pt,tmargin=117pt]{geometry}
\usepackage{hyperref}
\usepackage{pst-node}
\usepackage{tikz-cd}
\usepackage{tikz}
\usetikzlibrary{calc}
\usepackage[mathscr]{euscript}
\usepackage[numbers]{natbib}
\usepackage[toc,page]{appendix}

\author{Modjtaba Shokrian Zini, Zhenghan Wang and Xiao-Gang Wen}
\title{Pattern of Zeros}

\makeatletter
\newcommand\Author{Modjtaba Shokrian Zini, Zhenghan Wang and Xiao-Gang Wen}
\let\Title\@title
\def\ps@mystyle{%
      \let\@oddfoot\@empty\let\@evenfoot\@empty
      \def\@evenhead{\makebox[0pt][l]{\thepage}\hfill\Author\hfill}%
      \def\@oddhead{\hfill\Title\hfill\makebox[0pt][l]{\thepage}}%
      \let\@mkboth\markboth}
\makeatother
\makeatletter 
\g@addto@macro{\endabstract}{\@setabstract}
\newcommand{\authorfootnotes}{\renewcommand\thefootnote{\@fnsymbol\c@footnote}}%
\makeatother
\makeatletter
\renewcommand{\maketitle} 
{ \begingroup \vskip 10pt \begin{center} \large {\bf \@title}
	\vskip 10pt \large \@author \hskip 20pt \@date \end{center}
  \vskip 10pt \endgroup \setcounter{footnote}{0} }
\makeatother 

\newtheorem{thm}{Theorem}[section]

\newtheorem{cor}[thm]{Corollary}

\theoremstyle{definition}
\newtheorem{dfn}{Definition}
\theoremstyle{remark}
\newtheorem{rmk}{Remark}

\usepackage{url,times,wasysym,geometry,indentfirst,rotating,tikz}
\title{\LARGE{\bf
\textsc{Pattern of Zeros}}}

\begin{document}
\begin{center}
  \LARGE 
  \maketitle \par \bigskip

  \normalsize
  \authorfootnotes
  Modjtaba Shokrian Zini \footnote{shokrian@math.ucsb.edu}\textsuperscript{,\hyperref[1]{1}}, Zhenghan Wang \footnote{zhenghwa@microsoft.com;
  zhenghwa@math.ucsb.edu}\textsuperscript{,\hyperref[2]{2}} and Xiao-Gang Wen
  \footnote{xgwen@mit.edu}\textsuperscript{,\hyperref[3]{3}} \par \bigskip
\end{center}

\begin{abstract}
In this paper, we provide a set of definitions upon which one can prove in a rigorous way most of the main results achieved in the pattern of zeros classification of fractional quantum Hall states.
\end{abstract}
\section{Introduction}\label{intro}

Classifying quantum phases of matter is an interesting mathematical problem, but probably impossible without restrictions.  Recently, important progress towards some classifications has been made for one class of quantum phases of matter---topological phases of matter (TPMs).  The most important class of TPMs is fractional quantum Hall liquids (FQH) of electrons.  An elementary approach called patterns of zeros is initiated in \cite{PhysRevB.77.235108}.  The Pauli exclusion principle is essential to our understanding of electronic states of matter.  The pattern-of-zero approach is based on a quantified generalization of the Paul exclusion to clusters of electrons, which is essentially a poor-man's version of the powerful conformal field theory (CFT) approach to quantum Hall states.

FQH states are quantum phases of matter that exhibit patterns of long range entanglements with topological symmetry in their ground states. As a symmetry breaking state has only short range entanglement, the Landau symmetry breaking description of phases and phase transitions cannot be adequate for describing these new states of matter.  The FQH states have same symmetries while they have different topological orders, as their universal properties such as the degeneracy of ground states are stable under any small perturbation \cite{wen1990ground}. Indeed, the degenerate ground states are indistinguishable from each other under any local perturbation up to exponential accuracy \cite{bonderson2013quasi}. As outlined in \cite{wen2014pattern}, in order to classify FQH ground states by their characteristic topological data, there have been different approaches. 

The pattern of zeros approach \cite{PhysRevB.77.235108} is inspired by the CFT approach \cite{moore1991nonabelions}. Conformal blocks of some unitary rational CFTs turn out to be related to the \textit{Laughlin states} for $\nu=\frac{1}{m}$,
\begin{align}
\Phi_{\frac{1}{m}}(z_1,\ldots,z_N)=\prod_{i<j}(z_i-z_j)^m 
\end{align} and quasiparticle excitations. The same observation was later made for the parafermion states \cite{read1999beyond} and how CFT conformal blocks encode the way the FQH states goes to zero as a cluster of $k$ particles fuse. This is in fact the idea behind the pattern of zeros. This information on the order of zeros is enough to have a classification of FQH states and their quasiparticles in many instances \cite{wen2014pattern,PhysRevB.77.235108,lu2010non,wen2008topological}. As mentioned in \cite{wen2014pattern}, it should be noted that the CFT approach is not one-to-one and the pattern of zeros can also be thought of a classification of the CFTs giving FQH states.

The idea of pattern-of-zero approach comes from the cluster property of FQH states.  Suppose $\Psi(z_1,z_2,...,z_N)$ is a wave function of $N$ electrons at positions $z_i$ (electrons in FQH states are spinless effectively as they are spin polarized by the strong magnetic fields).  By the Pauli exclusion principle, $\Psi(z_1,z_2,...,z_N)\rightarrow 0$ when two different electrons at $z_i$ and $z_j$ approach each other.  Trial wave functions of FQH states have a cluster property and the wave functions take on simple forms when all $N$ electrons are divided into clusters \cite{read1999beyond} such as the Laughlin wave function.  So in a sense, pattern of zero states are clustered descendants of the Laughlin state.

The pattern of zeros approach tries to classify some symmetric complex functions by their order of zeros when fusion occurs. But as we are thinking of a quantum mechanical system of $N$ particles with $N \to \infty$ in the thermodynamical limit, we expect Hilbert spaces having these polynomials as ground states and a map between the low energy subspaces preserving the space of ground states, hence an algebraic relation between the polynomials. It is not clear yet what explicitly this relation is and it does not seem to be easy to obtain as it may require a diagonalization of the Hamiltonian. In fact, more importantly, we do not have yet a mathematical rigorous formulation of the Hilbert space and the Hamiltonian in general (although, see \cite{girvin1984formalism} on this issue for the lowest Landau levels).

As a result, we will not be able to provide a complete mathematical framework for this approach yet. Nevertheless, we will explore the advances one can make in this regard to put some of the main results of this classification program on a rigorous foundation.

\section{Notations and Conventions}\label{notations}

We work with variables representing positions of different particle types. Each particle of type $a$ is obtained through \lq\lq fusion" (clustering) of some $a_i$-type particles, where $\sum a_i=a$ and denoted by $z^{(a)}$. No superscript would mean a $1$-type particle.

Fusion is denoted by an arrow $(z^{(a_1)},\ldots,z^{(a_k)}) \rightarrow z^{(a)}$ meaning the weighted average $z^{(a)}= \frac{\sum a_iz^{(a_i)}}{a}$, which can be also called the \textit{center of mass}. Notice that different fusions (possibly done in many steps) of same variables give the same variable in the end with the same formula. A subscript may be needed for $z^{(a)}$ showing which $a$ variables are fused together, but mostly it would be clear from the context. Occasionally, we will have to use subscripts like $z^{(a)}_i$ to enumerate them and distinguish the different $a$-type particles. 

As a convention, whenever the word \textit{symmetric} is used, it always means \textit{symmetric between particles of the same type}. 

As an example for all the above, a symmetric polynomial $P(z^{(2)}_1,z^{(2)}_2,z_4,z^{(6)})$ is a polynomial where there are two $2$-type particles, one $1$-type particle, and one $6$-type particle. Exchanging $z^{(2)}_1$ and $z^{(2)}_2$ does not change the polynomial.

Finally, throughout the paper, we will work with polynomials which are \textit{derived} from one original polynomial in a specific way. Those polynomials, although being different, will be denoted by $P_\text{der}$. No confusion should arise as the variables that $P_\text{der}$ depends on will be clear.

\section{Overview of the previous definitions and results}\label{s2}

We will follow \cite{PhysRevB.77.235108} to provide the definitions and results. Assume a polynomial $P(z^{(b_1)}_1,\ldots,z^{(b_n)}_n,\ldots)$ where the number of variables are not specified. In order to take the fusion $(z^{(b_1)}_1,\ldots,z^{(b_a)}_a) \to z^{(k)}$ where $\sum_{1}^{a} b_i=k$ and get the corresponding wave function, replace $z^{(b_i)}_i=z^{(k)}+\lambda \xi_i$ where $\lambda$ is a scalar such that $\sum |b_i\xi_i|^2=1$ (also by definition $\sum b_i\xi_i = 0$). If the following decomposition exists
\begin{align}
P=\lambda^{S_k}R(\xi_1,\ldots,\xi_a)P_{\text{der}}(z^{(k)},z^{(b_{a+1})}_{a+1},\ldots,z^{(b_n)}_n)+ \text{h.o.t}
\end{align}
where h.o.t are the higher order terms in $\lambda$, $P_{\text{der}}$ is called the \textit{derived} polynomial for the fusion and the $S_k$ form a pattern of zeros for the polynomial. Of course, one has to consider the dependence of $S_k$ on the types of fused particles. But as for FQH states, polynomials all have at the start only $1-$type particles (representing bosons) and are symmetric, the definition of $S_k$ does not depend on the choice of variables to be fused and is in fact equal to the least total degree of the $k$ variables.

Further, being able to continue to carry fusions is important and $P_{\text{der}}$ should have a similar decomposition $R'P_{\text{der}}'$. Also different paths of fusions should give same derived polynomials (up to some scalar). This condition is called the \textit{Unique Fusion Condition} (UFC). 

The most simple fusion is that of two particles $(z^{(a)},z^{(b)}) \to z^{(a+b)}$. In this case, the power of $\lambda$ in that fusion is denoted by $D_{a,b}$. For a symmetric translation invariant UFC polynomial, it has been argued \cite{PhysRevB.77.235108} that
\begin{align}
D_{a,b}=S_{a+b}-S_a-S_b.
\end{align}
Further, the following \textit{concavity} inequality holds
\begin{align}
D_{a,b+c} \ge D_{a,b}+D_{a,c} \leftrightarrow S_{a+b+c}+S_a+S_b+S_c \ge S_{a+b}+S_{b+c}+S_{c+a}.
\end{align}

As mentioned before, in order to classify FQH states, a \textit{sequence} of polynomials satisfying certain conditions should be considered. This means symmetric translation invariant UFC polynomials $P_k(z_1,\ldots,z_{N_k})$ with $N_k$ many variables having a filling fraction $\nu=\frac{n}{m}$. This requirement is expressed as $\lim_{k\to \infty}\frac{N_k}{d_k}=\nu$, where $d_k$ is the highest degree of $z_i$ in $P_k$. Let us call such a sequence a $\nu$-\textit{sequence}.

The most important concept is the $n$-Cluster Form ($n$-CF) condition. Satisfying $n$-CF should be thought of  \cite{PhysRevB.77.235108} as having a periodic boson occupation in the orbitals: for every $n$ skip of bosons, there is $m$ skip of orbitals. 

As an example, consider the \textit{Laughlin state} for $\nu=\frac{1}{m}$,
\begin{align}
\Phi_{\frac{1}{m}}(z_1,\ldots,z_N)=\prod_{i<j}(z_i-z_j)^m.
\end{align}
This turns out to be a sequence satisfying all conditions and the ``$n$-CF''. One idea behind $n$-CF is that all FQH states should have as their eventual \textit{descendant} the $\nu$-Laughlin state $\Phi_\nu(z_1,\ldots,z_N)=\prod_{i<j}(z_i-z_j)^{\nu^{-1}}$; there exists a function $G_k$ such that,
\begin{align}
P_k=G_k\Phi_\nu, \forall k
\end{align}
and further, for any fusion of $n$ variables, $G_{k,\text{der}}$ is expected to not depend on $z^{(n)}$. Note that although $G_k$ and $\Phi_\nu$ are multi-valued functions but their product is not. This ensures that any fusion of $n$ variables gives a derived polynomial of the form $P_{\text{der}}=\prod_{i>n}(z^{n}-z_i)^mQ$ where $Q$ is some polynomial not depending on $z^{(n)}$. This condition turns out to be quite restrictive on the pattern of zeros. 

We list all main results \cite{PhysRevB.77.235108} for a UFC and $n$-CF $\nu$-sequence pattern of zeros $(\frac{n}{m};S_1,\ldots,S_n)$:
\begin{enumerate}
\item $S_{a+b}-S_a-S_b \ge 0$.
\item $S_{a+b+c}+S_a+S_b+S_c \ge S_{a+b}+S_{b+c}+S_{c+a}$.
\item $S_{2a}$ is even.
\item $mn$ even.
\item $2S_n \equiv 0 \pmod{n}$.
\item $S_{3a}-S_a$ even.
\item $S_{a+kn}=S_a+kS_n+kma+\frac{k(k-1)mn}{2}$.
\end{enumerate}

In the next sections, we will make definitions and investigate all the above claims mathematically. All claims, (with the exception of (e,f)) will be proven in \hyperref[s5]{section 5} after each concept is analyzed in a more general setting in \hyperref[s3]{section 3} and \hyperref[s4]{4}. Arguments will be made for (e,f) in \hyperref[s6]{section 6} and \hyperref[s3]{3} (respectively) that further assumptions are needed and the current conditions are likely not enough.

\section{Revision of the definitions and properties of $S_a$}\label{s3}

We start with some definitions and try to make clear the relation between $S_a$ and the least total degree.
\begin{dfn} \label{dfn1}
\textbf{(Prederived polynomial)} Given $P(z^{(b_1)}_1,\ldots,z^{(b_n)}_n)$, replace $z^{(b_i)}_i$ by $z^{(k)}+\lambda \xi_i$ where $z^{(k)}=\frac{\sum_{1}^a b_iz^{(b_i)}_i }{k}$, where $\lambda$ is a scalar such  that $\sum |b_i\xi_i|^2=1$. Expand $P$ in terms of $\lambda$ as
\begin{align}
P=\sum \lambda^m Q^m(\xi,z^{(k)},\ldots),
\end{align}
where $\xi=(\xi_1,\ldots,\xi_a)$. Let $s$ be the lowest power of $\lambda$ in that expansion. Then $Q^{s}$ is called the \textit{prederived polynomial}. 
\end{dfn}

If $Q^{s}$ is decomposable as a product of a polynomial in $\xi$ and a polynomial in $(z^{(k)},z^{(b_{a+1})}_{a+1},\ldots,z^{(b_n)}_n)$, ``the'' derived polynomial (up to some scalar) is obtained:

\begin{dfn} \label{dfn2}
\textbf{(Derived polynomial)} With the same settings as \hyperref[dfn1]{\textbf{Definition 1}}, if $Q^{s}=R(\xi)P_{\text{der}}(z^{(k)},z^{(b_{a+1})}_{a+1},\ldots,z^{(b_n)}_n)$, then $P_{\text{der}}$ is called the \textit{derived polynomial}.
\end{dfn}

A \textbf{fusion process} is a sequence of fusions where at each step the derived polynomial is taken and then the next fusion is performed (this assumes the existence of derived polynomial at each step). The final set of variables is called the \textbf{final configuration}. This allows us to define
    
\begin{dfn} \label{dfn3}
\textbf{(sUFC or Symmetric, Unique Fusion Condition)} $P$ satisfies UFC if any two fusion processes on the variables of $P$ having the same final configuration give the same derived polynomial. It satisfies sUFC if it is also symmetric.
\end{dfn}

The fusion process was defined as one fusion at a time. As multiple fusions could happen at the same time, stronger conditions are needed.
\begin{dfn} \label{dfn4}
\textbf{(Prederived polynomial revisited)}
Consider $u$ fusions occurring at the same time with their corresponding $\lambda_1,\ldots,\lambda_u$, and $\xi$s called $\xi_1,\ldots,\xi_u$ and their center of masses $z^{(k_1)},\ldots,z^{(k_u)}$. The polynomial coefficient of the least total power of $\lambda_i$ (given by replacing all $\lambda_i$s with $\lambda$) is called the prederived polynomial for this multiple fusion.
\end{dfn}
\begin{dfn} \label{dfn5}
\textbf{(Derived polynomial revisited)} Same settings as in \hyperref[dfn4]{\textbf{Definition 4}}, if the prederived polynomial is decomposable as $$R(\xi_1,\ldots,\xi_u)P_{\text{der}}(z^{(k_1)},\ldots,z^{(k_u)},\ldots),$$
then $P_{\text{der}}$ is the derived polynomial after those fusions.
\end{dfn}

The notion of fusion process can be similarly redefined and as a result, the notion of UFC becomes a stronger restriction on the polynomial. Note that if $P$ satisfies sUFC, then it follows from the definition that any derived polynomial from $P$ also satisfies sUFC.

In \hyperref[dfn1]{\textbf{Definition 1}}, the polynomial was expanded as $P=\sum \lambda^m Q^m$, but sometimes, this is not the best form of writing $P$ to understand fusion. As translation invariant polynomials will be assumed in the future, $P$ has to be written in a form which shows the ``parts" which are not translation invariant. Take the following algebra 
\begin{align}\label{eq7}
\mathbb{C}[z^{(b_2)}_2-z^{(b_1)}_1, \ldots , z^{(b_i)}_i-z^{(\sum_{1}^{i-1} b_j)}, \ldots,z^{(b_{a})}_{a}-z^{(\sum_{1}^{a-1} b_j)}, z^{(\sum_{1}^a b_j)}],
\end{align}
where the fusion is $(z^{(b_1)}_1,\ldots, z^{(b_{i-1})}_{i-1})\to z^{(\sum_{1}^{i-1} b_j)}$. Notice that the $a$ expressions above are algebraically independent as the $a$ variables $(z^{(b_1)}_1,\ldots, z^{(b_{a})}_a)$ can be obtained linearly. This implies that a homogeneous polynomial in each of those $a$ variables is also a homogeneous polynomial in terms of the other with the same total degree.

The first $a-1$ expressions are translation invariant as $(z^{(b_1)}_1,\ldots, z^{(b_{a})}_a) \rightarrow (z^{(b_1)}_1+c,\ldots, z^{(b_{a})}_a+c)$ for any $c \in \mathbb{C}$ leaves them invariant. The last expression, i.e. $z^{(\sum_{1}^a b_j)}$, changes to $z^{(\sum_{1}^a b_j)}+c$. Further, the first $a-1$ variables in the algebra \hyperref[eq7]{(7)}, can be rewritten using $\xi$ and $\lambda$, and a homogeneous expression with degree $m$ is also homogeneous in $\xi$ with degree $m$ multiplied by $\lambda^m$.

$P$ can now be expressed in the desired form. Write $P=\sum L^t$ where $L^t$ are homogeneous of degree $t$ in the variables that are going to be fused together, in this case $(z^{(b_1)}_1,\ldots,z^{(b_a)}_a)$. Then $L^t$ is also homogeneous of degree $t$ in $(z^{(b_2)}_2-z^{(b_1)}_1, \ldots , z^{(b_i)}_i-z^{(\sum_{1}^{i-1} b_j)}, \ldots,z^{(b_{a})}_{a}-z^{(\sum_{1}^{a-1} b_j)}, z^{(\sum_{1}^a b_j)})$, which for simplicity shall be replaced with $(y_1,\ldots,y_{a-1},z^{(k)})$, where $k=\sum_{1}^a b_j$. We have
\begin{align}\label{eq2}
L^t=\sum_{i=0}^{t_{\text{max}}} B^{t,t-i}_{z^{(b_{a+1})}_{a+1},\ldots,z^{(b_n)}_n}(y_1,\ldots,y_{a-1})(z^{(k)})^i,
\end{align}
where $B^{t,t-i}$ is a homogeneous polynomial of degree $t-i$ in $(y_1,\ldots,y_{a-1})$ and depends on the other variables $(z^{(b_{a+1})}_{a+1},\ldots,z^{(b_n)}_n)$ in the subscript. This subscript is dropped whenever it is clear from the context what the rest of the variables are. As it is clear from the summation, $t_{\text{max}}$ denotes the highest degree of $z^{(k)}$ in $L^t$.

Let us explore, in general, the properties of $S_{z^{(b_1)}_1,\ldots,z^{(b_a)}_a}$ (defined as the \textit{least total degree} of the subscript variables) using the $L^t$ expansion of the polynomial. 

Each $B^{t,t-i}$ has a factor of $\lambda^{t-i}$. With the same notations in \hyperref[dfn1]{\textbf{Definition 1}}, let us take all pairs $(t,i)$ for which $t-i=s$. As $s$ is the minimum power of $\lambda$, it follows $t-i \geq s$ for all $(t,i)$. Hence, if a pair $(t,i)$ satisfies $t-i=s$, it implies $i=t_{\text{max}}$. Therefore, if an expression from $L^t$ is in the prederived polynomial, it must be $B^{t,t-t_{\text{max}}}(z^{(k)})^{t_{\text{max}}}$:
\begin{align}
\lambda^s Q^s=\sum_{(t,t_{\text{max}}) \in S} B^{t,t-t_{\text{max}}}(z^{(k)})^{t_{\text{max}}}
\end{align}
where those pairs $(t,t_{\text{max}})$ for which $t-t_{\text{max}}=s$ form the set $S$. Notice the right side of the above equation is an expansion with respect to $z^{(k)}$ as there is only one pair $(t,t_{\text{max}})$ for any $t$. Assuming a derived polynomial,
\begin{align}\label{eq10}
\lambda^s R(\xi)P_{\text{der}}(z^{(k)},z^{(b_{a+1})}_{a+1},\ldots,z^{(b_n)}_n) = \sum_{(t,i) \in S} B^{t,t-i}(z^{(k)})^{i}.
\end{align}
Expanding $P_{\text{der}}=\sum P_{\text{der}}^i(z^{(b_{a+1})}_{a+1},\ldots,z^{(b_n)}_n) (z^{(k)})^i$ and recalling that $R$ is a polynomial in terms of $y_i$s, we get:
\begin{align}\label{eq11}
\lambda^s R(\xi) P_{\text{der}}^i = B^{t,t-i}, \ \ \ \forall (t,i) \in S.
\end{align}
$L^t$ is the homogeneous parts of $P$ in $z^{(b_1)}_1,\ldots,z^{(b_a)}_a$ with total degree $t$. Hence as $t \geq t-t_{\text{max}} \geq s$, $s$ is at \textbf{most} the \textbf{least} total degree of $z^{(b_1)}_1,\ldots,z^{(b_a)}_a$ in $P$, i.e. $S_{z^{(b_1)}_1,\ldots,z^{(b_a)}_a}$. 

If equality happens, then $t_{\text{max}}=0$ for $t=s$, which means that there is no power of $z^{(k)}$ in $L^{S_{z^{(b_1)}_1,\ldots,z^{(b_a)}_a}}$. Therefore, $L^{S_{z^{(b_1)}_1,\ldots,z^{(b_a)}_a}}$ appears in $Q^s$ as the ``constant" term with respect to $z^{(k)}$, so $z^{(k)} \nmid Q^s$. On the other hand, if $z^{(k)} \nmid Q^s$, then the constant term with respect to $z^{(k)}$ is some $B^{t,t-t_{\text{max}}}$ for $t=s$ where $t_{\text{max}}=0$, giving us the same conclusion $S_{z^{(b_1)}_1,\ldots,z^{(b_a)}_a}=s$.
\begin{thm}\label{thm3.1}
With the same notations as \hyperref[dfn1]{\textbf{Definition 1}}, for all polynomials $P$, we have $s \leq S_{z^{(b_1)}_1,\ldots,z^{(b_a)}_a}$. Equality happens if and only if $z^{(k)} \nmid Q^s$, in which case the constant term with respect to $z^{(k)}$ in $Q^s$ is $$L^{S_{z^{(b_1)}_1,\ldots,z^{(b_a)}_a}}=B^{S_{z^{(b_1)}_1,\ldots,z^{(b_a)}_a},S_{z^{(b_1)}_1,\ldots,z^{(b_a)}_a}}.$$
\end{thm}

Assuming a derived polynomial, \hyperref[eq11]{(11)} gives us
\begin{thm}\label{thm3.2}
For all polynomials $P$ having a derived polynomial, we have $s \leq S_{z^{(b_1)}_1,\ldots,z^{(b_a)}_a}$. Equality happens, if and only if $z^{(k)} \nmid P_{\text{der}}$ in which case the constant term with respect to $z^{(k)}$ in $R(\xi)P_{\text{der}}$ is 
$$R(\xi)P_{\text{der}}^0=L^{S_{z^{(b_1)}_1,\ldots,z^{(b_a)}_a}}=B^{S_{z^{(b_1)}_1,\ldots,z^{(b_a)}_a},S_{z^{(b_1)}_1,\ldots,z^{(b_a)}_a}}$$
\end{thm}

As explained in \hyperref[s2]{section 2}, the equality has to happen for the polynomials of interest. It turns out that translation invariant symmetry is the symmetry that will enforce the equality.

Suppose that $Q^s$ is the prederived polynomial in \hyperref[dfn1]{\textbf{Definition 1}}. If $P$ is translation invariant,
$$P(z^{(b_1)}_1,\ldots,z^{(b_n)}_n)=P(z^{(b_1)}_1+c,\ldots,z^{(b_n)}_n+c),$$
which after expansion with respect to $\lambda$ gives
$$Q^m(\xi,z^{(k)},z^{(b_{a+1})}_{a+1},\ldots,z^{(b_n)}_n)= Q^m(\xi,z^{(k)}+c,z^{(b_{a+1})}_{a+1}+c,\ldots,z^{(b_n)}_n+c),$$
where $\xi$ does not change by translation of the variables due to its definition. Therefore, all $Q^m$'s including $Q^s$, are translation invariant, implying that
$$0 \not \equiv Q^m(\xi,z^{(k)},z^{(b_{a+1})}_{a+1},\ldots,z^{(b_n)}_n)=Q^m(\xi,0,z^{(b_{a+1})}_{a+1}-z^{(k)},\ldots,z^{(b_n)}_n-z^{(k)}).$$
Hence, $z^{(k)} \nmid Q^m$. If the derived polynomial exists, as $R$ depends only on $\xi$ which does not change when a translation is performed, $P_{\text{der}}$ is also translation invariant. Therefore, similarly $z^{(k)} \nmid P_{\text{der}}$.
\begin{cor}\label{cor3.3}
For a translation invariant polynomial $P$, we have $s=S_{z^{(b_1)}_1,\ldots,z^{(b_a)}_a}$. The same corollary holds for all the derived polynomials of $P$ obtained from any fusion process.
\end{cor}
\begin{rmk}
In the special case of one variable as the final configuration of a fusion process a constant polynomial is obtained (as it has to be translation invariant).
\end{rmk}
We would like to comment on the claim (f) in \hyperref[s2]{section 2}. The argument made in \cite{PhysRevB.77.235108} for the claim does not use anything more than the sUFC translation invariant condition. Even assuming homogeneity for the polynomial, there is a counter-example to the claim. It should be noted that the polynomial is not part of an $n$-CF $\nu$-sequence and as numerical evidence is in support of claim (f), it is reasonable to expect that $n$-CF and the relationship between the polynomials in the $\nu-$sequence should play a role.

Consider the symmetric homogeneous translation invariant polynomial $$P=(x-(x+y+z)/3)(y-(x+y+z)/3)(z-(x+y+z)/3).$$ 
This polynomial is obviously UFC (indeed, any three variable translation invariant polynomial can be easily seen to be UFC). Expanding gives
$$\frac{2x^3}{27}-\frac{x^2y}{9}-\frac{x^2z}{9}-\frac{xy^2}{9}+\frac{4xyz}{9}-\frac{xz^2}{9}+\frac{2y^3}{27}-\frac{y^2z}{9}-\frac{yz^2}{9}+\frac{2z^3}{27}.$$
While $S_1=S_2=0$, we have $S_3=3$. Hence the claim (f) being $S_{3a}-S_{a} \equiv 0 \pmod 2$ does not hold.

\section{Revision of $D_{a,b}$ and its properties}\label{s4}

We will start with a general polynomial and consider restrictions along the path; until all the \textit{necessary} requirements are found for a notion of $D_{a,b}$ satisfying the desired properties outlined in \hyperref[s2]{section 2}. In a \textit{two-fusion}, using the algebra in \hyperref[eq7]{(7)}:
\begin{align}\label{eq12}
P(z^{(a)},z^{(b)},z^{(c)},\ldots)=(z^{(a)}-z^{(b)})^{D_{a,b}} \Big( P^{a,b}_{\text{der}}(z^{(a+b)},z^{(c)},\ldots)
\end{align}
$$+\sum_{i>0}(z^{(a)}-z^{(b)})^{i}P^{a,b}_i(z^{(a+b)},z^{(c)},\ldots) \Big ),$$
where $D_{a,b}$ is the smallest power of that term. Note that for a two-fusion, there is always a derived polynomial. Next, consider the \textit{three-fusion} 
$$F_{a,b,c}: (z^{(a)},z^{(b)},z^{(c)},\ldots) \rightarrow (z^{(a+b+c)},\ldots)$$
for which $P$ is assumed to have a derived polynomial. In this fusion, instead of using \hyperref[eq7]{(7)}, the algebra
\begin{align}\label{eq13}
\mathbb{C}[z^{(a)}-z^{(b)},z^{(a)}-z^{(c)},z^{(a+b+c)},\ldots]
\end{align}
will be used. Notice $\mathbb{C}[z^{(a)}-z^{(b)},z^{(a)}-z^{(c)}]=\mathbb{C}[\lambda,\xi]$ and $z^{(a)}-z^{(b)},z^{(a)}-z^{(c)},z^{(b)}-z^{(c)}$ give each one power of $\lambda$. Using the notations in \hyperref[eq12]{(12)}, they divide $P$ with order $D_{a,b},D_{a,c}$ and $D_{b,c}$. Hence,
\begin{align}\label{eq14}
P=(z^{(a)}-z^{(b)})^{D_{a,b}}(z^{(a)}-z^{(c)})^{D_{a,c}}(z^{(b)}-z^{(c)})^{D_{b,c}} 
\end{align}
$$\Big ( R(z^{(a)}-z^{(b)},z^{(a)}-z^{(c)})P^{a,b,c}_{\text{der}}(z^{(a+b+c)},\ldots)+ \text{h.o.t. in $\lambda$} \Big ),$$
where $R$ is a homogeneous polynomial in $(z^{(a)}-z^{(b)},z^{(a)}-z^{(c)})$ with total degree $s_{a,b,c}-D_{a,b}-D_{a,c}-D_{b,c}$, and $s_{a,b,c}$ is the lowest power of $\lambda$ in the three-fusion. Next, we consider the three-fusion carried out in the following way:
$$F_{a,(b,c)}=(z^{(a)},z^{(b)},z^{(c)},\ldots) \rightarrow (z^{(a)},z^{(b+c)},\ldots) \rightarrow (z^{(a+b+c)},\ldots).$$
UFC is not assumed, so the derived polynomial from this fusion process could be different. For the first step:
\begin{align}\label{eq15}
P(z^{(a)},z^{(b)},z^{(c)},\ldots)=(z^{(b)}-z^{(c)})^{D_{b,c}}\Big (P^{b,c}_{\text{der}}(z^{(a)},z^{(b+c)},\ldots)
\end{align}
$$+\sum_{i>0}(z^{(b)}-z^{(c)})^{i}P^{b,c}_i(z^{(a)},z^{(b+c)},\ldots) \Big ).$$

As for the second step, consider the expansion of $P^{b,c}_{\text{der}}(z^{(a)},z^{(b+c)},\ldots)$ and also the expansion of $P^{b,c}_i(z^{(a)},z^{(b+c)},\ldots)$ with respect to $(z^{(a)}-z^{(b+c)})$:
\begin{align}\label{eq16}
P^{b,c}_{\text{der}}(z^{(a)},z^{(b+c)},\ldots)=(z^{(a)}-z^{(b+c)})^{D_{a,b+c}}P^{a,(b,c)}_{\text{der}}(z^{(a+b+c)},\ldots)
\end{align}
$$+ \text{h.o.t in $(z^{(a)}-z^{(b+c)})$},$$
\begin{align}\label{eq17}
P^{b,c}_i(z^{(a)},z^{(b+c)},\ldots)=(z^{(a)}-z^{(b+c)})^{D_{a,b+c}-i-l_{\text{min},i}}P^{a,(b,c)}_i(z^{(a+b+c)},\ldots)
\end{align}
$$+ \text{h.o.t in $(z^{(a)}-z^{(b+c)})$},$$
where $D_{a,b+c},D_{a,b+c}-i-l_{\text{min},i}$ are the maximum powers of $(z^{(a)}-z^{(b+c)})$ dividing $P^{b,c}_{\text{der}}$ and $P^{b,c}_i$, respectively.

Assume $l_{\text{min},i} > 0$ for some $i>0$. Consider the \textbf{elementary basis} for the homogeneous polynomials in $\lambda,\xi$ after expanding $P$ using \hyperref[eq16]{(16)} and \hyperref[eq17]{(17)}. The homogeneous parts in $z^{(b)}-z^{(c)},z^{(a)}-z^{(b+c)}$ (which expands the same algebra as $z^{(a)}-z^{(b)},z^{(a)}-z^{(c)}$) with total degree $s_{a,b,c}$ should have all coefficient equal to $P_{\text{der}}(z^{(a+b+c)},\ldots)$ up to some scalar. Thus,
\begin{gather}\label{eq18}
\scalebox{0.9}{$
(z^{(a)}-z^{(b)})^{D_{a,b}}(z^{(a)}-z^{(c)})^{D_{a,c}}(z^{(b)}-z^{(c)})^{D_{b,c}}R(z^{(a)}-z^{(b)},z^{(a)}-z^{(c)})=$}
\end{gather}
\begin{gather*}
\scalebox{0.9}{$(z^{(b)}-z^{(c)})^{D_{b,c}}\Big (\sum\limits_{\{ i| D_{a,b+c}-D_{b,c}-l_{\text{min},i}=s_{a,b,c} \}}c_i(z^{(b)}-z^{(c)})^{i}(z^{(a)}-z^{(b+c)})^{D_{a,b+c}-i-l_{\text{min},i}} \Big)$}
\end{gather*}
where as explained, the equality comes from the fact that $P^{a,b,c}_{\text{der}}(z^{(a+b+c)},\ldots)= P^{a,(b,c)}_i(z^{(a+b+c)},\ldots)$ for all those $i \in \{ i| D_{a,b+c}-D_{b,c}-l_{\text{min},i}=s_{a,b,c} \}$. 

If for all $i>0$ we have $l_{\text{min},i} \le 0$, repeating the same procedure,
$$
(z^{(a)}-z^{(b)})^{D_{a,b}}(z^{(a)}-z^{(c)})^{D_{a,c}}(z^{(b)}-z^{(c)})^{D_{b,c}}R(z^{(a)}-z^{(b)},z^{(a)}-z^{(c)})=
$$
\begin{align}\label{eq19}
(z^{(b)}-z^{(c)})^{D_{b,c}}\Big (\sum_{i \geq 0} c_i(z^{(b)}-z^{(c)})^{i}(z^{(a)}-z^{(b+c)})^{D_{a,b+c}-i} \Big).
\end{align}
The difference from the previous case is the contribution from $P^{b,c}_{\text{der}}$ to the right side of the above equation and so $c_0 \neq 0$, i.e. the highest power of $(z^{(a)}-z^{(b+c)})$ in the above equation is $D_{a,b+c}$.
Both sides in \hyperref[eq19]{(19)} are homogeneous polynomials in variables that are related. Denote $(z^{(a)}-z^{(b)},z^{(a)}-z^{(c)})$ by $(x,y)$ and $(z^{(b)}-z^{(c)},z^{(a)}-z^{(b+c)})=(w,t)$. The following hold
$$x=t-\frac{c}{b+c}w, \ \ y=t+\frac{b}{b+c}w$$
$$\implies y-x=w$$
For both cases (\hyperref[eq18]{(18)} and \hyperref[eq19]{(19)}), we have (respectively):
\begin{align}\label{eq20}
(x)^{D_{a,b}}(y)^{D_{a,c}}(y-x)^{D_{b,c}}R(x,y)=
\end{align}
$$w^{D_{b,c}} \Big (\sum_{ \{ i| D_{a,b+c}-D_{b,c}-l_{\text{min},i}=s_{a,b,c} \}} c_i(w)^{i}(t)^{D_{a,b+c}-i-l_{\text{min},i}} \Big),$$
\begin{align}\label{eq21}
(x)^{D_{a,b}}(y)^{D_{a,c}}(y-x)^{D_{b,c}}R(x,y)=w^{D_{b,c}} \Big (\sum_{i \geq 0} c_i(w)^{i}(t)^{D_{a,b+c}-i} \Big),
\end{align}
In the case of \hyperref[eq20]{(20)},
$$
(t-\frac{c}{b+c}w)^{D_{a,b}}(t+\frac{b}{b+c}w)^{D_{b,c}}w^{D_{b,c}}R(t-\frac{c}{b+c}w,t+\frac{b}{b+c}w)=$$
\begin{align}\label{eq22}
w^{D_{b,c}} \Big (\sum_{ \{ i| D_{a,b+c}-D_{b,c}-l_{\text{min},i}=s_{a,b,c} \}} c_i(w)^{i}(t)^{D_{a,b+c}-i-l_{\text{min},i}} \Big).
\end{align}
On the left side of the equation, the highest degree of $t$ is at least $D_{a,b}+D_{a,c}$. On the right side, the highest power of $t$ is at most the minimum of $D_{a,b+c}-i-l_{\text{min},i} \le D_{a,b+c}-l_{\text{min},i} < D_{a,b+c}$. Notice that the previous inequalities could have been an equality, i.e. $D_{a,b+c}$ is the highest power of $t$, if the case \hyperref[eq21]{(21)} held true as $c_0 \neq 0$. Thus, in both cases,
\begin{align}\label{eq23}
D_{a,b+c} \geq D_{a,b}+D_{a,c}.
\end{align}
Recall that the only assumption on the starting polynomial $P(z^{(a)},$ $z^{(b)},z^{(c)},\ldots)$ was the existence of a derived polynomial for the fusion $(z^{(a)},z^{(b)},z^{(c)},\ldots) \rightarrow (z^{(a+b+c)},\ldots)$. In fact, even this condition can be seen to be unnecessary and \textbf{any} polynomial $P$ satisfies \hyperref[eq23]{(23)}. The change that needs to be made is in \hyperref[eq22]{(22)}, where one has to replace the scalar $c_i$'s by polynomials in $(z^{(a+b+c)},\ldots)$ and replace $R$ with a summation of the form
$$\sum R_i(t-\frac{c}{b+c}w,t+\frac{b}{b+c}w)Q_i(z^{(a+b+c)},\ldots)$$
where $R_i$'s are homogeneous polynomials in $(x,y)=(t-\frac{c}{b+c}w,t+\frac{b}{b+c}w)$. 
\begin{thm}\label{thm4.1} \textbf{(Concavity condition)}
For any polynomial $P(z^{(a)},z^{(b)},z^{(c)},\ldots)$,
$$D_{a,b+c} \geq D_{a,c}+D_{a,b},$$
where $D_{a,b}$ is the power of $(z^{(a)}-z^{(b)})$ dividing $P$ (similarly for $D_{a,c}$), and $D_{a,b+c}$ is the power of $(z^{(a)}-z^{(b+c)})$ dividing the derived polynomial from the fusion of $z^{(b)},z^{(c)}$.
\end{thm}
Suppose that $P$ is symmetric and $a=b=c$. In this case, \hyperref[eq12]{(12)} gives $D_{a,a} \equiv 0 \pmod 2$ by considering the transposition $z^{(a)}_1 \leftrightarrow z^{(a)}_2$.
\begin{thm}\label{thm4.2}\textbf{(Evenness condition)}
For every symmetric polynomial $P$ with two $a$-type particles, we have $D_{a,a} \equiv 0 \pmod 2$.
\end{thm}

Assume a sUFC translation invariant polynomials $P(z_1,\ldots,z_n)$ where all particles are of type one. Occasionally, we will also suppose that $P$ is homogeneous if needed to. As the polynomial is translational invariant and symmetric, \hyperref[cor3.3]{\textbf{Corollary 3.3}} gives $S_a=$ (minimum total power of $a$ variables) as a well-defined notion. Note that $S_1=0$ due to translation invariance.

We wish to prove $D_{a,b}=S_{a+b}-S_a-S_b$. $D_{a,b}$ for all $a,b$ needs to be defined consistently. Previously, this notation was used when an $a$-type and a $b$-type particle existed from the start. One way to define this notion is to consider the derived polynomial of the fusion $(z_1,\ldots,z_n) \rightarrow (z^{(a)},z^{(b)},z_{a+b+1},\ldots,z_n)$. Then define $D_{a,b}$ as the order of $(z^{(a)}-z^{(b)})$ dividing the derived polynomial $P_{\text{der}}(z^{(a)},z^{(b)},z_{a+b+1},\ldots,z_n)$. But this needs to be consistent among \textbf{all} derived polynomials of $P$ which have an $a$-type and a $b$-type particle. One can construct an example where this notion is not well-defined. Consider the sUFC translation invariant homogeneous polynomial:
$$P(z_1,z_2,z_3,z_4)=\sum_{k<l}\frac{\prod_{i<j}(z_i-z_j)^2}{(z_k-z_l)^2}.$$
To check that the polynomial is UFC, it only needs to be done for the three-fusion and the fusion $(z_1,z_2,z_3,z_4)\to (z^{(2)}_1,z_3,z_4) \to (z^{(3)}_1,z_4)$ and both give the derived polynomial $(z^{(3)}-z_4)^6$. It is clear that $D_{1,1}=0$ due to the division performed. Further, one can see that $P_{\text{der}}(z^{(2)}_1,z_3,z_4)=(z_1^{(2)}-z_3)^4(z_1^{(2)}-z_4)^4(z_3-z_4)^2$. The order of $(z_3-z_4)$ is $2$. Thus, there can be no consistent notion of $D_{1,1}$ in this case. This condition must be imposed on the polynomials:
\begin{dfn}\label{dfn6}(\textbf{Unique Local Condition (ULC)})
Consider a sUFC translation invariant polynomial $P=P(z_1,\ldots,z_n)$ where all particles are of of type $1$. If any of the derived polynomials which has an $a$-type particle and a $b$-type particle is divisible by $(z^{(a)}-z^{(b)})$ with a maximum order of $D_{a,b}$, then $P$ satisfies ULC.
\end{dfn}
Similar to claim (f) analyzed in the previous section, it could be the case that the yet undiscovered relation between the polynomials in an $n$-CF $\nu$-sequence provides the above condition automatically.

To prove $D_{a,b}=S_{a+b}-S_a-S_b$, let us only assume translation invariance and symmetry and the existence of $P_{\text{der}}(z^{(a)},z_{a+1},\ldots)$. Let $D_{a,1}$ to be the power of $z^{(a)}-z_{a+1}$ dividing the derived polynomial. The inequality
$$D_{a,1} \geq S_{a+1}-S_{a}-S_1=S_{a+1}-S_{a},$$ 
can then be proven. Take the expansion:
$$P= R(z_2-z_1,z_3-z^{(2)},\ldots,z_a-z^{(a-1)})P_{\text{der}}(z^{(a)},z_{a+1},\ldots) + \text{h.o.t in $\lambda$},$$
where $R$ is a homogeneous polynomial of degree $S_a$. Taking the fusion
\begin{align}\label{eq24}
P_{\text{der}}(z^{(a)},z_{a+1},\ldots)=(z_{a+1}-z^{(a)})^{D_{a,1}}P_{\text{der}}(z^{(a+1)},z_{a+2},\ldots)
\end{align}
$$+ \text{h.o.t in $(z_{a+1}-z^{(a)})$}.$$
As $R$ is multiplied by $(z_{a+1}-z^{(a)})^{D_{a,1}}$, the order of $\lambda$ becomes $D_{a,1}+S_a$ (when the fusion of the first $a+1$ variables is taken). The least power of $\lambda$ in $P$ is $S_{a+1}$ implying
$$D_{a,1}+S_a \geq S_{a+1} \implies D_{a,1} \geq S_{a+1}-S_{a}.$$

Assuming $P$ is sUFC and homogeneous as well, $P_{\text{der}}(z^{(a+1)},z_{a+2},\ldots)$ has total degree $\text{deg}(P)-S_{a}-D_{a,1}$. This derived polynomial is same as the derived polynomial obtained by \textbf{directly} computing the fusion of the first $a+1$ variables which means it has total degree $\text{deg}(P)-S_{a+1}$. Thus,
$$\text{deg}(P)-S_{a+1}=\text{deg}(P)-S_{a}-D_{a,1} \implies D_{a,1}=S_{a+1}-S_{a}.$$

Finally, assume that $P$ also satisfies ULC. Starting with $P_{\text{der}}(z^{(a)},\ldots)$, inductively fuse the variables $(z_{a+1},\ldots,z_{a+b})$. In the first step, the derived polynomial $P_{\text{der}}(z^{(a)},z^{(2)},z_{a+3},\ldots)$ is obtained with degree $\text{deg}(P)-S_{a}-D_{1,1}$. The next step is the fusion of $z^{(2)},z_{a+3}$ giving the derived polynomial $P_{\text{der}}(z^{(a)},z^{(3)},\ldots)$ with degree $\text{deg}(P)-S_{a}-D_{1,1}-D_{2,1}$. Repeating this procedure, it is easy to see that the polynomial $P_{\text{der}}(z^{(a)},z^{(b)},z_{a+b+1},\ldots)$ has degree:
$$\text{deg}(P)-S_{a}-\sum_{i=1}^{b-1}D_{i,1}$$
But as $D_{i,1}= S_{i+1}-S_{i}$, the degree is
$$\text{deg}(P)-S_{a}-S_{b}.$$
At the final step, fusing $z^{(a)},z^{(b)}$ gives $P_{\text{der}}(z^{(a+b)},z_{a+b+1},\ldots)$ with degree
$$\text{deg}(P)-S_{a}-S_{b}-D_{a,b},$$
as the order of $(z^{(a)}-z^{(b)})$ dividing $P_{\text{der}}(z^{(a)},z^{(b)},z_{a+b+1},\ldots)$ is $D_{a,b}$. But one can also get $P_{\text{der}}(z^{(a+b)},z_{a+b+1},\ldots)$ directly from one single fusion and the degree should be $\text{deg}(P)-S_{a+b}$. This implies
$$\text{deg}(P)-S_{a+b}=\text{deg}(P)-S_{a}-S_{b}-D_{a,b} \implies D_{a,b}=S_{a+b}-S_{a}-S_{b}.$$

All conditions except homogeneity were shown to be necessary to obtain the equality above. If homogeneity is not given, one can think of a counterexample which is sUFC translation invariant and satisfies ULC while $D_{2,1}>S_{3}-S_2$. 

Note that the argument for $D_{i,1} \geq S_{i+1}-S_{i}$ can be applied for non-symmetric translation invariant polynomials where we will have to replace $D_{i,1},S_{i+1},S_{i}$ by $D_{z_{i+1},z^{(i)}},S_{z_1,\ldots,z_{i+1}},S_{z_1,\ldots,z_i}$ which will be dependent on which variables are being fused. Assume a $3$ variable polynomial $P$ which is translation invariant and satisfies $D_{z_3,z^{(2)}} > S_{z_1,z_2,z_3}-S_{z_1,z_2}$. The ``product symmetrization" of this polynomial is
$$SymmP(z_1,z_2,z_3) = \prod_{\sigma \in S_3} P_\sigma=\prod_{\sigma \in S_3} P(\sigma(z_1),\sigma(z_2),\sigma(z_3))$$
$SymmP$ is symmetric and also translation invariant as $P$ is. It is UFC as any three variable translation invariant polynomial is UFC. $SymmP$ also satisfies ULC, as any symmetric translation invariant polynomial with three variables does so.

Finally, the pattern of zeros of $SymmP$ comes from the pattern of zeros of $P$:
$$S^{SymmP}_2 = \sum_\sigma S^{P_\sigma}_{z_1,z_2}= \sum_\sigma S^{P_\sigma}_{z_2,z_3}= \sum_\sigma S^{P_\sigma}_{z_3,z_1},$$
and similarly
$$S^{SymmP}_3 = \sum_\sigma S^{P_\sigma}_{z_1,z_2,z_3},$$
which is obvious due to the definition of the above values as the lowest total power of a number of variables.  Further, the derived polynomials of $SymmP$ are the product of the derived polynomials of $P_\sigma$. And,
$$D^{SymmP}_{1,1}=\sum_\sigma D^{P_\sigma}_{z_1,z_2}=\sum_\sigma D^{P_\sigma}_{z_2,z_3}= \sum_\sigma D^{P_\sigma}_{z_3,z_1}.$$
Similarly, 
$$D^{SymmP}_{2,1}=\sum_\sigma D^{P_\sigma}_{z^{(2)}_1,z_3}=\sum_\sigma D^{P_\sigma}_{z^{(2)}_2,z_1}= \sum_\sigma D^{P_\sigma}_{z^{(2)}_3,z_2}.$$
As each of the terms in $\sum_\sigma D^{P_\sigma}_{z^{(2)}_1,z_3}$ is at least $S^{P_\sigma}_{z_1,z_2,z_3}-S^{P_\sigma}_{z_1,z_2}$, $$D^{SymmP}_{2,1} \geq S^{SymmP}_3-S^{SymmP}_2.$$
Hence, if for one instance of $\sigma$, one has $$D^{P_\sigma}_{z^{(2)}_1,z_3} > S^{P_\sigma}_{z_1,z_2,z_3}-S^{P_\sigma}_{z_1,z_2}$$ 
then we get a strict inequality for $SymmP$. Now for demonstrating the necessity of homogeneity, the polynomial $P(z_1,z_2,z_3)=y^3+x^2$ where $y=z_3-\frac{z_1+z_2}{2}$ and $x=z_1-z_2$ can be used as an example. This polynomial is translation invariant and if $z_1,z_2$ are fused, the derived polynomial gives $y^3$, so $D^P_{z^{(2)}_1,z_3}=3$ while $S^P_{z_1,z_2,z_3}=2$ and $S^P_{z_1,z_2}=0$.

\section{$\nu$-sequence and $n$-CF}\label{s5}

In this section, we will investigate the claims made at the end of \hyperref[s2]{section 2} on an $n$-CF $\nu$-sequence of polynomials. In the same spirit of \cite{PhysRevB.77.235108}, let us define
\begin{dfn}\label{dfn7}\textbf{($\nu$-sequence)} A sequence of sUFC translation invariant ULC homogeneous polynomials $\{P_k(z_1,\ldots,z_{N_k})\}$ is called a $\nu$-sequence if 
\begin{itemize}
\item for a positive $\nu \in \mathbb{Q}$ we have $\lim_{k \rightarrow \infty} \frac{N_k}{d_k}=\nu$, where $d_k$ is the maximum degree of $z_1$ in $P_k$,
\item $S_a$ for all $P_k$ for which $N_k \geq a$ is the same, so that a sequence of pattern of zeros $\{S_a\}$ can be associated to the sequence of polynomials.
\end{itemize}
\end{dfn}
The sequence $\{ D_{a,b} \}$ can similarly be associated to a $\nu$-sequence satisfying $D_{a,b}=S_{a+b}-S_a-S_b$. Due to concavity,
$$D_{a,b+c} \geq D_{a,b}+D_{a,c} \implies$$
\begin{align}\label{eq25}
S_{a+b+c}+S_{a}+S_{b}+S_{c} \geq S_{a+b}+S_{b+c}+S_{c+a}.
\end{align}
Also, the evenness condition holds
\begin{align}\label{eq26}
D_{a,a} \equiv 0 \pmod 2 \implies S_{2a} \equiv 0 \pmod 2.
\end{align}
Finally, $S_{N_k}=\text{deg}(P_k)$ as $P_k$s are homogeneous. All claims (a)-(c) in \hyperref[s2]{section 2} are therefore already established rigorously. As for claim (f), it was demonstrated in \hyperref[s3]{section 3} that its proof likely requires a deeper understanding of the relations between the polynomials in a $\nu$-sequence.

The definition of $n$-CF is also done similar to the definition in \hyperref[s2]{section 2}, but by simply avoiding the multi-valued issue for the functions.
\begin{dfn}\label{dfn8}(\textbf{$n$-Cluster Form condition ($n$-CF)})
A $\nu$-sequence with $\nu=\frac{n}{m}$ and $N_k=kn$, satisfies the $n$-CF if
\begin{itemize}
    \item $P_k^n=G_k \Phi_{\frac{1}{m}}$ for all $k$, where $\Phi_{\frac{1}{m}}$ is $\prod_{1 \leq i<j \leq nk} (z_i-z_j)^m$ and $G_k$ is a meromorphic function $\frac{P_k^n}{\Phi_{\frac{1}{m}}}$.
    \item The derived polynomial $$G_{k,\text{der}}(\ldots,z^{(n)},\ldots):=\frac{P_{k,\text{der}}(\ldots,z^{(n)},\ldots)^n}{\Phi_{\frac{1}{m},\text{der}}(\ldots,z^{(n)},\ldots)}$$
    is independent of any $n$-type particle $z^{(n)}$, where $P_{k,\text{der}},\Phi_{\frac{1}{m},\text{der}}$ are arbitrary derived polynomials from the same fusion process having an $n$-type variable in the final configuration. 
\end{itemize}
\end{dfn} 

Recall that the derived polynomial of a product is the product of the derived polynomial and its pattern of zeros of is also obtained by adding up the pattern of zeros. Hence, $P_{k,\text{der}}(\ldots,z^{(n)},\ldots)^n$ is the derived polynomial of $P_k^n$ from the same fusion process. Next, more or less the same arguments made in \cite{PhysRevB.77.235108} will be used to prove the recursive relation made in claim (g). First, we prove
\begin{align}\label{eq27}
nD_{a,b}=mab  \ , \ \text{if} \ b \equiv 0 \pmod n.
\end{align}
Take a fusion process with a final configuration $(\ldots,z^{(n)},z^{(a)},\ldots)$. Fusing the two variables $z^{(n)}$ and $z^{(a)}$,
$$G_{k,\text{der}}(\ldots,z^{(n)},z^{(a)},\ldots)=$$
$$\frac{(z^{(n)}-z^{(a)})^{nD_{n,a}}P_{k,\text{der}}(\ldots,z^{(n+a)},\ldots)^n+ \text{h.o.t in $(z^{(n)}-z^{(a)})$}}{(z^{(n)}-z^{(a)})^{man}\Phi_{\frac{1}{m},\text{der}}(\ldots,z^{(n+a)},\ldots)+ \text{h.o.t in $(z^{(n)}-z^{(a)})$}}=$$
$$(z^{(n)}-z^{(a)})^{nD_{a,n}-man}\frac{P_{k,\text{der}}(\ldots,z^{(n+a)},\ldots)^n+\text{h.o.t in $(z^{(n)}-z^{(a)})$}}{\Phi_{\frac{1}{m},\text{der}}(\ldots,z^{(n+a)},\ldots)+\text{h.o.t in $(z^{(n)}-z^{(a)})$}}.$$

As $G_{k,\text{der}}$ is independent of $z^{(n)}$, we deduce $nD_{a,n}-man=0$. Also, by definition
$$G_{k,\text{der}}(\ldots,z^{(n+a)},\ldots)=$$
$$\lim_{z^{(n)} \rightarrow z^{(a)} \text{or} \lambda \rightarrow 0}\frac{P_{k,\text{der}}(\ldots,z^{(n+a)},\ldots)^n+\text{h.o.t in $(z^{(n)}-z^{(a)})$}}{\Phi_{\frac{1}{m},\text{der}}(\ldots,z^{(n+a)},\ldots)+\text{h.o.t in $(z^{(n)}-z^{(a)})$}}$$
The term inside the limit is $G_{k,\text{der}}(\ldots,z^{(n)},z^{(a)},\ldots)$ which is independent of $z^{(n)}$. Assume that $a=n$ and denote the variable $z^{(n)}_2$, not the same as the other $n$-type particle $z^{(n)}_1=z^{(n)}$. One side is $G_{k,\text{der}}(\ldots,z^{(n+n)},\ldots)$ as the other side is a limit of $G_{k,\text{der}}(\ldots,z^{(n)}_1,z^{(n)}_2,\ldots)$, and is independent on both those variables. Therefore, $G_{k,\text{der}}(\ldots,z^{(n+n)},\ldots)$ is independent of $z^{(n+n)}$. As the argument applies to any derived $G_{k,\text{der}}$, any such function is independent of $2n$-type variables. 

Repeating the argument above for $2n$ instead of $n$ and arguing inductively, \hyperref[eq27]{(27)} is proved. 

Another identity is needed to prove the recursive relation:
\begin{align}\label{eq28}
D_{a,b+n}=D_{a,b}+am.
\end{align}
Using again the limit above, one can conclude that as long as $z^{(n)}=z^{(a)}$,
$$G_{k,\text{der}}(\ldots,z^{(n+a)},\ldots)= G_{k,\text{der}}(\ldots,z^{(a)},\ldots).$$ 
To prove \hyperref[eq28]{(28)}, we consider the following equality
$$G_{k,\text{der}}(\ldots,z^{(a)},z^{(n+b)},\ldots)=G_{k,\text{der}}(\ldots,z^{(a)},z^{(b)},\ldots).$$
when $z^{(n)}=z^{(b)}$. Taking the fusion of $z^{(a)},z^{(n+b)}$ on one side and the fusion of $z^{(a)},z^{(b)}$ on the other side:
\begin{gather}\label{eq29}
\scalebox{0.97}{$(z^{(a)}-z^{(n+b)})^{nD_{a,b+n}-ma(b+n)}\frac{P_{k,\text{der}}(\ldots,z^{(a+b+n)},\ldots)^n+\text{h.o.t in $(z^{(a)}-z^{(b+n)})$}}{\Phi_{\frac{1}{m},\text{der}}(\ldots,z^{(a+b+n)},\ldots)+\text{h.o.t in $(z^{(a)}-z^{(b+n)})$}}$}
\end{gather}
\begin{gather*}
\scalebox{0.97}{$=(z^{(a)}-z^{(b)})^{nD_{a,b}-mab}\frac{P_{k,\text{der}}(\ldots,z^{(a+b)},\ldots)^n+\text{h.o.t in $(z^{(a)}-z^{(b)})$}}{\Phi_{\frac{1}{m},\text{der}}(\ldots,z^{(a+b)},\ldots)+\text{h.o.t in $(z^{(a)}-z^{(b)})$}}.$}
\end{gather*}

As $z^{(n)}=z^{(b)}$, i.e. $z^{(b+n)}=z^{(b)}$, comparing the order of $(z^{(a)}-z^{(b)})$ on both sides gives \hyperref[eq28]{(28)}.

The recursive relation for the pattern of zeroes can be proved by combining \hyperref[eq27]{(27)} and \hyperref[eq28]{(28)}. For $a<n$,
$$D_{a,kn}=mak \  \text{from (27)}, D_{a,b+kn}= D_{a,b}+kam \ \text{from (28)}.$$
The first equation gives
$$S_{a+kn}=S_{kn}+S_{a}+kam.$$
To find $S_{kn}$, one can use $S_{kn}=\sum_{1}^{kn-1}D_{1,i}$ and break this summation into $k$ parts for $i$ modulo $n$ and apply $D_{1,b+kn}=D_{1,b}+kma$ as follows
$$S_{kn}= k (\sum_{1}^{n-1} D_{1,i}) + \sum_{1}^{k-1} mn(i)=kS_n+\frac{mnk(k-1)}{2}.$$
Thus, claim (g) is proved:
\begin{align}\label{eq30}
S_{a+kn}=S_{a}+kma+kS_n+\frac{mnk(k-1)}{2}.
\end{align}

Notice \hyperref[eq27]{(27)}, for $a=b=n$, gives $D_{n,n}=mn$. From the evenness condition, $mn$ is even, which demonstrates claim (d). 

The last claim to be discussed is (e): $2S_n \equiv 0 \pmod{n}$. For an $n$-CF $\nu$-sequence, following \cite{PhysRevB.77.235108}, one has to first define the angular momentum of a particle $z^{(a)}$. By considering the $\mathfrak{su}(2)$ action on the polynomials \cite[see Appendix A]{PhysRevB.77.235108}, this was found to be $J_a=aJ-S_a$ where $J$ is half of $d_k$, the highest degree of $z_i$ in $P_k$. Then, by the recursive formula and the fact that the total angular momentum, i.e. $J_{\text{tot}}=N_kJ-S_{N_k}$ must be zero due to translation invariance \cite[see III.D]{PhysRevB.77.235108}, one obtains
$$2J=\frac{2S_{n}}{n}+m(\frac{N_k}{n}-1).$$
As $2J$ and $\frac{N_k}{n}$ are integers, claim (e) follows: $2S_n \equiv 0 \pmod{n}$.

Although almost all the above arguments can be made precise using very much the same ideas, the part of which does not seem to follow mathematical precision is the statement that $J_{\text{tot}}=0$ for such polynomials. In \cite{PhysRevB.77.235108}, only properties of the polynomial itself are used to state the claim, while one can easily find an sUFC translation invariant ULC homogeneous polynomial which does not satisfy $J_{\text{tot}}=0$. Consider for example the same polynomial we used for claim (f):
$$P=(x-(x+y+z)/3)(y-(x+y+z)/3)(z-(x+y+z)/3).$$ 
This gives $J_{\text{tot}}=3\times \frac{3}{2}-3 \neq 0$. Hence, there is need for a better understanding of what the definition of FQH ground state sequence is. In the next section, we analyze the significance of the \textbf{assumption} of $J_{\text{tot}}=0$ on an $n$-CF $\nu$-sequence and how $n$-CF itself can be ``derived" from this assumption.

\section{angular momentum $J_a=aJ_1-S_a$ and $n$-AMS polynomials}\label{s6}

Although the right definition of angular momentum requires an action of $\mathfrak{su}(2)$ on the polynomials, we refer to \cite[see Appendix A]{PhysRevB.77.235108} as the arguments are mathematically precise. In this section, we ``define" $J_a=aJ-S_a$ as the angular momentum of $z^{(a)}$ and explore its relationship with the highest degree of $z^{(a)}$ in $P_{\text{der}}(z^{(a)},\ldots)$.

Let $P$ be sUFC translation invariant. Recall that $S_a$, the least total power of $a$ variables, is equal to the least total power of $\lambda$ when fusing $a$ variables. Define the quantity $J_a := aJ-S_a$, where $J$ is half of the highest degree of $z_i$. 
\begin{thm}\label{thm6.1}
For any sUFC translation invariant polynomial $P(z_1,\ldots,z_N)$,
$$2J_a \geq d_a$$
where $d_a$ is  the highest degree of $z^{(a)}$ in $P_{\text{der}}(z^{(a)},z_{a+1},\ldots,z_N)$.
\end{thm} 
\begin{proof}
Proceeding by induction, for $a=1$, the equality follows from the definitions. To prove the theorem for the fusion of $a=a$ variables, we fuse $z^{(a-1)}$ and $z_{a}$:
\begin{gather}\label{eq31}
\scalebox{0.95}{$P_{\text{der}}(z^{(a-1)},z_{a},\ldots,z_N)=(z^{(a-1)}-z_a)^{D_{a-1,1}}P_{\text{der}}(z^{(a)},z_{a+1},\ldots,z_N)+ \ldots .$}
\end{gather}
Recall
\begin{align}\label{eq32}
D_{a-1,1} \geq S_{a}-S_{a-1},
\end{align}
as proven in \hyperref[s4]{section 4}. Let $d'$ be the highest degree of $z_a$ in $P_{\text{der}}(z^{(a-1)},z_{a},\ldots,z_N)$ which is at most $2J_1=2J$. We want to show that
$$2J_a=2J_{a-1}+2J_{1}-2S_a+2S_{a-1} \geq d_{a-1}+d'-2D_{a-1,1} \geq d_{a}.$$
Except for the right inequality, all the above are already given using the induction hypothesis and \hyperref[eq32]{(32)}. Define
\begin{gather}\label{eq33}
\scalebox{0.9}{$M(z^{(a)},z^{(a-1)}-z_a,z_{a+1},\ldots,z_N)=P_{\text{der}}(z^{(a)},z_{a+1},\ldots,z_N)+ (z^{(a-1)}-z_a)T$}
\end{gather}
where $T$ is some polynomial, to get $P_{\text{der}}(z^{(a-1)},z_{a},\ldots,z_N)=$
\begin{align}\label{eq34}
(z^{(a-1)}-z_a)^{D_{a-1,1}}M(z^{(a)},z^{(a-1)}-z_a,z_{a+1},\ldots,z_N).
\end{align}
Let the highest degree of $z^{(a-1)}$ and $z_a$ in $M$ be $m_{a-1}$ and $m_1$, respectively. From \hyperref[eq34]{(34)}, by counting the highest degree of $z^{(a-1)}$ and $z_a$ in two ways, 
$$m_{a-1} + D_{a-1,1} = d_{a-1},$$
$$m_1 +D_{a-1,1} = d'.$$
Adding up gives $m_{a-1}+m_{1} = d'+d_{a-1}-2D_{a-1,1}$. Therefore,it remains to show that
\begin{align}\label{eq35}
d_{a} \leq m_{a-1}+m_{1}.
\end{align}
Writing $M$ using the $L^t$ form,
\begin{align}\label{36}
M= \sum_{d_{\text{min}}}^{d_{\text{max}}} L^t,
\end{align}
where $L^t$ are the homogeneous ``parts" of $M$ with total degree $t$ in $z^{(a)},z^{(a-1)}-z_a$ or equivalently in $z^{(a-1)},z_a$ (notice $\mathbb{C}[z^{(a)},z^{(a-1)}-z_a]=\mathbb{C}[z^{(a-1)},z_a]$). Also, $d_{\text{min}}$ is the smallest total degree of the homogeneous parts and $d_{\text{min}}$ is the biggest. 

Notice that the polynomial $P_{\text{der}}(z^{(a)},z_{a+1},\ldots,z_N)$ does not depend on the variable $(z^{(a-1)}-z_a)$. Therefore, the highest degree of $z^{(a)}$ in $P_{\text{der}}$ comes from a term in $L^t$ for some $t=t'$, which is independent of $(z^{(a-1)}-z_a)$. In particular, one can conclude that $d_a \leq t'$ for some $t=t'$. But the maximum of $t$ is by definition $d_{\text{max}}$. Thus $d_a \leq d_{\text{max}}$. We also know that $L^{d_{\text{max}}}$ is the homogeneous part of $M$ in $z^{(a-1)},z_a$ with total degree $d_{\text{max}}$. So the highest \textbf{total} degree of $z^{(a-1)},z_a$ in $M$ is by definition $d_{\text{max}}$. Finally, since $m_{a-1},m_1$ are the highest degrees of $z^{(a-1)},z_a$ in $M$,
$$d_a \leq t' \leq d_{\text{max}} \leq m_{a-1}+m_1.$$
This finishes the proof.
\end{proof}
\begin{rmk}
As derived polynomials with an $a$-type variable are descendants of $P_{\text{der}}(z^{(a)},z_{a+1},\ldots)$ (up to some permutation of variables due to symmetry), the result above is true for all derived polynomials having some $a$-type variable. 
\end{rmk}
\begin{rmk}
Assume that $P$ is also homogeneous and satisfies ULC. Taking the derived polynomial $P_{\text{der}}(z^{(a)},z^{(N-a)})$, the following important inequality holds
\begin{align}\label{eq37}
J_N=NJ-S_N=J_{N-a}+J_{a}-(S_N-S_{N-a}-S_{a})
\end{align}
$$=J_{N-a}+J_{a}-D_{N-a,a} \geq 0,$$
as the highest degree of $z^{(N)}$ in $P_{\text{der}}(z^{(N)})$ is indeed $0$ (i.e. it is constant), since $P$ is translation invariant.
\end{rmk}
\begin{dfn}\label{dfn9}
(\textbf{Angular Momentum Symmetry (AMS)}) Let $P$ be an sUFC translation invariant ULC homogeneous polynomial. Then $P$ satisfies AMS if $J_N=0$.
\end{dfn}
The terminology is explained in \hyperref[eq39]{(39)} as it is equivalent to $J_N=0$ for the case $a=0$ in that equation. As $P$ satisfies ULC and is homogeneous,
\begin{align}
P_{\text{der}}(z^{(a)},z^{(N-a)})=c(z^{(a)}-z^{(N-a)})^{D_{N-a,a}}.
\end{align}
Indeed, $P_{\text{der}}$ must be homogeneous of degree $\text{deg}(P)-S_a-S_{N-a}=D_{N-a,a}$ and also divisible by $(z^{(a)}-z^{(N-a)})^{D_{N-a,a}}$ by the definition of $D_{N-a,a}$. \hyperref[thm6.1]{\textbf{Theorem 6.1}} gives
$$2J_{N-a} \geq D_{N-a,a} \ \ , \ \ 2J_a \geq D_{N-a,a},$$
as the highest power of both $z^{(a)}$ and $z^{(N-a)}$ in $P_{\text{der}}(z^{(a)},z^{(N-a)})$ is exactly $D_{N-a,a}$. Therefore, in order to have equality in \hyperref[eq37]{(37)}, the following has to hold 
\begin{align}\label{eq39}
J_{N-a}=\frac{1}{2}D_{N-a,a}=J_a.
\end{align}
This in turn implies that $2J_a$ is equal to the highest degree of $z^{(a)}$ in any derived polynomial depending on an $a$-type particle. Indeed, by the remark after \hyperref[thm6.1]{\textbf{Theorem 6.1}}, $z^{(a)}$ has maximum degree $D_{N-a,a}$ in $P_{\text{der}}(z^{(a)},z^{(N-a)})$, and any derived polynomial having an $a$-type variable has this polynomial as its \textit{descendant}, and the \textit{ancestor} of all derived polynomials having an $a$-type variable satisfies $D_{N-a,a} \le d_a \le 2J_a=D_{N-a,a}$.

We can actually get more out of the equality in \hyperref[eq39]{(39)}. By the definition of $J_a$,
\begin{align}\label{eq40}
aJ-S_a=(N-a)J-S_{N-a} \ \ \ , \ \ \ J_N=0 \implies J=\frac{S_N}{N},
\end{align}
\begin{align}\label{eq41}
\implies S_{a}=S_{N-a}-(N-2a)J=S_{N-a}-\frac{N-2a}{N}S_N.
\end{align}
In particular, as $2J$ is an integer,
\begin{align}\label{eq42}
2J=\frac{2S_N}{N} \in \mathbb{Z} \implies 2S_N \equiv 0 \pmod N.
\end{align}
We note that \hyperref[eq41]{(41)} is supported by numerical evidence for $n$-CF $\nu$-sequences as mentioned \cite[see (48)]{PhysRevB.77.235108}. Let us summarize all the above in the following theorem:
\begin{thm}\label{thm6.2}
If $P$ satisfies AMS, then
\begin{itemize}
    \item $J_N=0$ or equivalently $J=\frac{S_N}{N}$ which implies that $J_a=J_{N-a}=\frac{1}{2}D_{N-a,a}$.
    \item $S_a=S_{N-a}-\frac{N-2a}{N}S_n$.
    \item $J_a$ is half of the highest degree of $z^{(a)}$ in any derived polynomial depending on an $a$-type particle.
\end{itemize}
\end{thm}
How does this relate to $n$-CF?
\begin{dfn}\label{dfn10}
\textbf{($n$-AMS polynomial)} an $n$-AMS polynomial $P$ is one that satisfies AMS with $N$ variables such that $n|N$ and $J_n=\frac{1}{2}D_{n,1}(N-n)$.
\end{dfn}
This definition be viewed as a definition for an ``$n$-CF polynomial". One can define it for a \textit{single} polynomial by requiring all the conditions of the definition of $n$-CF except those that are related of being part of a sequence. Notice $D_{n,1}$ plays the role of $m$ in $n$-CF definition as $D_{n,1}=S_{n+1}-S_n=m$ in \hyperref[eq30]{(30)}. Therefore, we shall use $m$ instead of $D_{n,1}$.

Given an $n$-AMS polynomial $P(z_1,\ldots,z_N)$, define 
$$G:=\frac{P^n}{\Phi_{\frac{1}{m}}}.$$
By showing that any derived $G_{\text{der}}:=\frac{P^n_{\text{der}}}{\Phi_{\frac{1}{m},\text{der}}}$ is independent of any $n$-type particle, one can use all the arguments in the previous section to deduce all the identities of $n$-CF like \hyperref[eq30]{(30)} just for this polynomial. 

We have 
$$P_{\text{der}}(z^{(n)},z_{n+1},\ldots,z_N)=\prod_{i=n+1}^{N}(z^{(n)}-z_i)^mQ(z^{(n)},z_{n+1},\ldots,z_N).$$
But since $P$ is AMS, the highest degree of $z^{(n)}$ in $P_{\text{der}}$ is $2J_n=m(N-n)$. This implies that $Q$ can not depend on $z^{(n)}$ and also
$$G_{\text{der}}=\frac{P_{\text{der}}(z^{(n)},z_{n+1},\ldots,z_N)^n}{\Phi_{\frac{1}{m},\text{der}}(z^{(n)},z_{n+1},\ldots,z_N)}$$
does not depend on $z^{(n)}$ as $$\Phi_{\frac{1}{m},\text{der}}(z^{(n)},z_{n+1},\ldots,z_N)=\prod_{i=n+1}^{N}(z^{(n)}-z_i)^{mn} \prod_{n+1 \leq i < j \leq N} (z_i-z_j)^m.$$
Since a derived $G_{\text{der}}$ which \textit{could} depend on any $z^{(n)}$ is ultimately a descendant (modulo a permutation of variables) of
$$\frac{P_{\text{der}}(z^{(n)},z_{n+1},\ldots,z_N)^n}{\Phi_{\frac{1}{m},\text{der}}(z^{(n)},z_{n+1},\ldots,z_N)},$$
there is no derived $G_{\text{der}}$ depending on an $n$-type particle. Therefore, all identities like \hyperref[eq30]{(30)} hold for an $n$-AMS polynomial for all integers $\le N$ and $\equiv a \pmod{n}$.

Since we do not know if the polynomials in an $n$-CF $\nu$-sequence are AMS or not, we can not claim whether the two definitions are equivalent; although, it should be noted that using the same argument above, the highest degree of $z^{(n)}$ in $P_k$ is exactly $D_{n,1}(N_k-n)$ for an $n$-CF $\nu$-sequence, which could be smaller than $J_n$ by \hyperref[thm6.1]{\textbf{Theorem 6.1}}. 

Define $\nu$-sequence of $n$-AMS polynomials as a $\nu$-sequence of polynomials which are $n$-AMS but only with the requirement that a pattern of zeros can be meaningfully associated to the sequence (as described in \hyperref[dfn7]{\textbf{Definition 7}}). It turns out that the filling fraction is a consequence of this definition. 

\begin{thm}\label{thm6.3}
A $\nu$-sequence of $n$-AMS polynomials is also an $n$-CF $\nu$-sequence and if the polynomials in an $n$-CF $\nu$-sequence are AMS, then they are also $n$-AMS. Further, for a $\nu$-sequence of $n$-AMS polynomials:
$$S_{a}=S_{n-a}-\frac{n-2a}{n}S_n$$
\end{thm}
\begin{proof}
For the first part, we only need to prove that the filling fraction $\nu$ equals $\frac{n}{m}$ for $m=D_{n,1}$. We have to compute
$$\lim_{k \rightarrow \infty} \frac{N_k}{d_k}=\lim_{k \rightarrow \infty} \frac{N_k}{2J_{1,P_k}}=\lim_{k \rightarrow \infty} \frac{N_k}{\frac{2S_{N_k}}{N_k}}=\lim_{k \rightarrow \infty} \frac{N_k}{\frac{2S_{n}}{n}+m(\frac{N_k}{n}-1)},$$
where the fact that the polynomials are AMS ($\frac{S_{N_k}}{N_k}=J_{1,P_k}$) and satisfy the recursive relation in \hyperref[eq32]{(32)} has been used. By factoring out $N_k$ and taking the limit, we get $\nu=\frac{n}{m}$.

if the polynomials in an $n$-CF $\nu$-sequence are AMS, using the notations of the previous section, for every $P_k$, 
$$J_{N_k,P_k}=0 \implies S_{N_k}=N_kJ_{1,P_k} \implies J_{n,P_k}=nJ_{1,P_k}-S_n$$
$$=\frac{nS_{N_k}}{N_k}-S_n=\frac{m(N_k-n)}{2},$$
where in the last equality, \hyperref[eq30]{(30)} is used. As $m=D_{n,1}$, $P_k$ is $n$-AMS. 

Finally, the last statement for a $\nu$-sequence of $n$-AMS polynomials is an application of \hyperref[thm6.2]{\textbf{Theorem 6.2}} and \hyperref[eq32]{(32)} on any $P_k$ with $N=N_k$ variables.
\end{proof}

As the ultimate goal is to classify FQH states, first one needs to define sequence of polynomials with the right properties. Imposing more conditions as long as we are confident of them being ``physical" conditions should not be something to avoid even if it is not clear whether redundant ones are imposed. As the definition of $J_a$ is the angular momentum of the particle corresponding to $z^{(a)}$ in the FQH state, $J_a$ must be half of the highest degree of $z^{(a)}$ in the polynomial describing the FQH state. The AMS condition is equivalent to that when we consider an $n$-CF $\nu$-sequence. Therefore, It makes sense to consider AMS $n$-CF $\nu$-sequence to classify FQH states, even though we do not know whether this condition is redundant or not.

\bibliographystyle{apa}
\bibliography{main}

\address{\textsuperscript{1\label{1}}Dept of Mathematics, University of California,
Santa Barbara, CA 93106-6105, U.S.A.}
\address{\textsuperscript{2\label{2}}Microsoft Station Q and Dept of Mathematics, University of California,
Santa Barbara, CA 93106-6105, U.S.A.} 
\address{\textsuperscript{3\label{3}}Department of Physics, Massachusetts Institute of Technology, Cambridge, MA 02139, U.S.A.}

\end{document}